\documentclass[12pt]{article}
\usepackage[a4paper,
            bindingoffset=0.2in,
            left=1.2in,
            right=1.2in,
            top=1.55in,
            bottom=1.55in,
            footskip=.5in]{geometry}
\usepackage{graphicx}
\usepackage{amsmath}
\usepackage{amssymb}
\usepackage{amsthm}
\usepackage{bm}
\usepackage{natbib}
\usepackage{dsfont}
\usepackage{mathtools}

\begin{document}


\setlength{\baselineskip}{24pt}
\begin{center}
{\textbf{\large An Equilibrium Analysis of the Arad-Rubinstein Game*
}}
\medskip


\qquad{Christian Ewerhart$^\dagger$} \qquad \qquad {Stanisław Kaźmierowski$^\ddagger$}

\smallskip
\today
\end{center}


\noindent \textbf{Abstract.} Colonel Blotto games with discrete strategy spaces effectively illustrate the intricate nature of multidimensional strategic reasoning. This paper studies the equilibrium set of such games where, in line with prior experimental work, the tie-breaking rule is allowed to be flexible. We begin by pointing out that equilibrium constructions known from the literature extend to our class of games. However, we also note that, irrespective of the tie-breaking rule, the equilibrium set is excessively large. Specifically, any pure strategy that allocates at most twice the fair share to each battlefield is used with positive probability in some equilibrium. Furthermore, refinements based on the elimination of weakly dominated strategies prove ineffective. To derive specific predictions amid this multiplicity, we compute strategies resulting from long-run adaptive learning. 


\smallskip

\thispagestyle{empty}
\bigskip

{\footnotesize \setlength{\baselineskip}{18pt} \noindent *) This paper has benefited from valuable comments and suggestions by Marcin Dziubiński, Dan Kovenock, Brian Roberson, and attendants of the session ``Contests I'' at the SAET Conference in Santiago de Chile. The second-named author thanks the Department of Economics of the University of Zurich for its hospitality. The work has been supported by the Polish National Science Center through grant 2018/29/B/ST6/00174. This manuscript reflects copy-edit suggestions put forth by ChatGPT 3.5.

\smallskip\noindent $\dagger$) Department of Economics, University of Zurich; christian.ewerhart@econ.uzh.ch.

\smallskip\noindent $\ddagger$) Department of Computer Science, University of Warsaw; s.kazmierowski@uw.edu.pl.

}

\newtheorem{definition}{Definition}
\newtheorem{theorem}{Theorem}
\newtheorem{proposition}{Proposition}
\newtheorem{observation}{Observation}
\newtheorem{lemma}{Lemma}
\newtheorem{example}{Example}
\newtheorem{assumption}{Assumption}
\newtheorem{corollary}{Corollary}
\newtheorem*{example2}{Example 2 (continued)}
\newtheorem*{lemma1}{Lemma 1 (extension)}
\makeatletter
\renewenvironment{proof}[1][\proofname]{%
  \par\pushQED{\qed}\normalfont%
  \topsep6\p@\@plus6\p@\relax
  \trivlist\item[\hskip\labelsep\bfseries#1\@addpunct{.}]%
  \ignorespaces
}{%
  \popQED\endtrivlist\@endpefalse
}

\newcommand{\pure}{s}
\newcommand{\aPure}{\pure^{A}}
\newcommand{\bPure}{\pure^{B}}
\newcommand{\purePair}{ \aPure, \bPure }
\newcommand{\pureProfile}{(\purePair)}
\newcommand{\purePayoff}{\mathrm{\pi}}
\newcommand{\mixedPayoff}{\mathrm{\pi}}
\newcommand{\dist}{\mathrm{\Delta}}

\newpage
\setlength{\baselineskip}{24pt}
\noindent\textbf{1. Preliminaries}
\medskip

\noindent\textit{1.1 Introduction} 

\noindent In a Colonel Blotto game, as envisaged by Borel (\citeyear{borel}), two adversaries are tasked with allocating their budgets of a resource secretly over a given set of battlefields, aiming to conquer as many battlefields as possible. In each battlefield, victory is awarded to the adversary allocating a higher amount of the resource, where a tie-breaking rule is invoked when both parties assign an equal amount. Applications of Colonel Blotto games extend beyond military conflicts to areas such as strategic marketing, electoral competition, innovation contests, and network security. Key contributions in the literature, including Roberson (2006) for continuous and Hart (2008) for discrete strategy spaces, assumed that the Colonel Blotto game is \textit{constant-sum at ties}, i.e., that any tied battlefield is ultimately conquered by one contestant or another, but never lost. This characteristic of the standard model ensures that the Colonel Blotto game is a two-person constant-sum game, significantly simplifying the equilibrium analysis. Notably, however, the assumption is not satisfied in experimental studies such as Arad and Rubinstein's (2012) investigation of multidimensional strategic reasoning. In their case, tied battlefields count as lost for both parties, rendering the standard equilibrium characterization invalid.\footnote{See Arad and Rubinstein (2012, p.~584) for their statement: ``We are not aware of any analysis of the Nash equilibria of our version of the game.''}

The present study investigates a class of Colonel Blotto games characterized by discrete strategy spaces and flexible tie-breaking rules. Our model nests both the constant-sum version of the game and its variant where tied battlefields generate no value. We present an approach that allows constructing Nash equilibria in this wider class of games. In particular, we find an equilibrium in the example considered by Arad and Rubinstein (2012). However, a caveat emerges. Specifically, our analysis reveals an excessively large equilibrium set within the model. 
Furthermore, attempts to narrow the set of equilibria down by applying the concept of weak dominance prove ineffective. To derive specific predictions nevertheless, we explore the implications of long-run adaptive learning.

\bigskip
\noindent\textit{1.2 Contribution} 

\noindent The present study contributes to the literature on finite Blotto games in two main ways. Firstly, we construct Nash equilibria in the analyzed class of games. For the constant-sum game, Hart~(\citeyear{hart}) demonstrated that marginal distributions on individual battlefields can be chosen to be essentially uniform. Consequently, assuming divisibility, players can partition the set of battlefields into pairs, evenly distribute their resources, and randomize uniformly within each pair. We observe that the precise form of the tie-breaking rule is less crucial, as the probability of a tie remains constant across bid levels used in equilibrium. We also point out that the loss of payoffs due to ties, when compared to the constant-sum version, is insufficient to make overbidding an attractive strategy. Based on these observations, we establish equilibria in the model with flexible tie-breaking, specifically in the Arad-Rubinstein game. In doing so, we identify conditions under which the equilibrium initially identified by Hart~(2008) persists in games with modified tie-breaking.

Our second main contribution is the observation that irrespective of the tie-breaking rule, the equilibrium set of games in the considered class is excessively large. More specifically, we show that \textit{any} pure strategy that does not put too many resources on any individual battlefield is part of
some equilibrium strategy. 
The idea of identifying pure strategies that are part of some equilibrium arises already in Tukey (\citeyear{tukey1949problem}), who termed those strategies as ``good.'' We find that any pure strategy that allocates at most twice the fair share of the budget to each battlefield is good in this sense.\footnote{Conversely, strategies that concentrate the resource on too few battlefields are never good.}
Moreover, none of these strategies can be eliminated by equilibrium refinements based on the elimination of weakly dominated strategies. As a remedy to this problem, we propose simulating long-run adaptive learning, which indeed yields definite predictions commensurate with experimental data.\footnote{The analysis also leads to a number of supplementary observations on the constant-sum model, which we decided to report upon in a separate section before the conclusion.}

\bigskip
\noindent\textit{1.3 Related literature} 

\noindent In his seminal paper on Colonel Blotto games, Borel (\citeyear{borel}) considered both continuous and finite strategy spaces. For the model with continuous strategy spaces, Roberson~(\citeyear{roberson}) characterized optimal marginal distributions and unique equilibrium payoffs. He also noted that in the continuous model, tie-breaking rules often do not matter (but may need modification to ensure existence). 

For the model with discrete strategy spaces, Hart (\citeyear{hart}) constructed equilibria not only in all cases with homogeneous endowments, but also for special cases with heterogeneous endowments. He used two auxiliary models. In the \textit{Colonel Lotto game}, players can be thought of as being restricted to mixed strategies that are invariant under arbitrary permutations of the set of battlefields. In the \textit{General Lotto game}, the budget constraint needs to be satisfied in expectation only. Solutions of the General Lotto game turn into solutions of the corresponding Colonel Lotto and Colonel Blotto games provided that those solutions are feasible, i.e., marginals can be derived from a joint distribution that satisfies the budget constraint with equality. 
Building on these concepts, Dziubiński~(\citeyear{dziubinski_1}, \citeyear{dziubinski_2}) characterized the set of optimal marginal distributions in the General Lotto game and, provided that the number of resources is divisible by the number of battlefields, also the set of optimal marginal distributions of Colonel Blotto games with discrete strategy spaces. Despite this progress, the general characterization of the equilibrium set in Colonel Blotto games with discrete strategy spaces has remained elusive. More recently, Liang et al.~(\citeyear{liang}) and Aspect and Ewerhart~(\citeyear{aspect2022colonel}) characterized equilibria in discrete Colonel Blotto games with two battlefields. Prior work has not attempted to characterize the equilibrium in a discrete Colonel Blotto game with a modified tie-breaking rule.\footnote{However, Rapoport and Almadoss (\citeyear{rapoport2000mixed}) and Dechenaux et al.\ (\citeyear{dechenaux2006caps}) have considered an all-pay auction with discrete bids and a modified tie-breaking rule.} 

Experimental tests of Blotto games not mentioned above include Avrahami and Kareev (\citeyear{avrahami2009weak}), Kohli et al.~(\citeyear{kohli2012colonel}), Chowdhury et al.~(\citeyear{chowdhury2013experimental}), Avrami et al.~(\citeyear{avrahami2014allocation}), and Montero et al.~(\citeyear{montero2016majoritarian}), among others.


\bigskip
\noindent\textit{1.4 Overview of the paper}

\noindent The structure of this paper is as follows. Section~2 introduces the model. In Section~3, we present an approach for constructing Nash equilibria within the analyzed class of games. Section~4 concerns the equilibrium set. Section~5 deals with refinements, while Section~6 reports on the simulation of adaptive learning utilizing a high-performance computing environment. Section~7 discusses extensions. Our supplementary findings regarding the standard model are compiled in Section 8. Section~9 concludes the paper. Technical proofs are provided in an appendix.

\newpage
\noindent \textbf{2. The model}
\medskip

\noindent \textit{2.1 Setup and notation}

\noindent Two players, denoted by ${A}$ and ${B}$, each allocate a total of $N \geq 1$ units of a resource over $K \geq 2$ battlefields.\footnote{In the excluded cases where $N=0$ or $K=1$, players have only one strategy. However, note that we allow for $K=2$, which is not necessarily trivial in our model.} Units of the resource are not divisible.
Hence, a \textit{pure strategy} of player $i$ is a vector
\begin{equation}
   s^i= \left( 
\begin{matrix}
s^i_1 \\ 
\vdots \\
s^i_K%
\end{matrix}%
\right)\text{,}\notag
\end{equation}
such that $s_i^k\in\{0,1,\ldots,N\}$ for every $k\in\{1,\ldots,K\}$, and 
\begin{equation}
 \sum^K_{k=1} s_k^i = N.\notag   
\end{equation}
The set of strategies for players $A$ and $B$ is identical and denoted by $S$.
As usual, we refer to the opponent of player $i$ by $-i$. For a given pure strategy profile $(s^{i}, s^{-i})  \in S \times S$, the payoff of player $i \in \{{A}, {B}\}$ in the Colonel Blotto game is defined by
\begin{equation}\label{eq:payoffs_blotto}
    \purePayoff^{i}\left(\pure^{i}, \pure^{-i} \right) = \sum^K_{k=1} \left( \mathds{1}_{s_k^{i} > s_k^{-i}} + \frac{\alpha}{2} \cdot \mathds{1}_{s_k^{i} = s_k^{-i}} \right),\notag
\end{equation}
where $\mathds{1}_{s_k^{i} > s_k^{-i}}$ equals one if player $i$'s bid in battlefield $k$ exceeds that of player $-i$, and zero otherwise, $\mathds{1}_{s_k^{i} = s_k^{-i}}$ equals one in the case of a tie on battlefield $k$, and zero otherwise, and $\alpha$ is a parameter. 
The departure from the standard model is the introduction of flexible tie-breaking, represented by $\alpha$. 
We call the two-player game with strategy sets and payoffs defined as above the \textit{Colonel Blotto game} $\mathcal{B}_{\alpha}\equiv\mathcal{B}_{\alpha}(N,K)$. 

\bigskip
\noindent \textit{2.2 Examples}

\noindent Below, we recall two examples of finite Colonel Blotto games that have been considered in the literature.

\begin{example} [\textbf{Hart, \citeyear{hart}}]
In the standard version of the Colonel Blotto game, $\mathcal{B}_{1}(N,K)$, player $i$'s payoff is defined as
\begin{equation}
    \purePayoff^{i}\left(\pure^{i}, \pure^{-i} \right)= \sum^K_{k=1} \left( \mathds{1}_{s_k^{i} > s_k^{-i}} + \frac{1}{2} \cdot \mathds{1}_{s_k^{i} = s_k^{-i}} \right).\notag
\end{equation}
\end{example}

\bigskip
\noindent The constant-sum setup underlying Example 1 has indeed been prevalent in the literature. Borel (\citeyear{borel}, p.~100) considered a particular case where $N=7$ and $K=3$, a solution of which has been offered by Hart (\citeyear{hart}). The setup in Example 1 has been tested as a symmetric control by Avrahami et al.~(\citeyear{avrahami2014allocation}) for $N\in \{16,24\}$ and $K=8$. 

\medskip
\begin{example}[\textbf{Arad and Rubinstein, \citeyear{arad_rubinstein}}]\label{ex:arad_rubinstein}
Two colonels are asked to distribute a total of $N=120$ units of the resource over a total of $K=6$ battlefields. The payoff of a player is the number of battlefields that she assigned strictly more resources than her opponent. Thus, the game is $\mathcal{B}_0(120,6)$, the common set of strategies is given as $S=\{s^i\in \{0,1,\ldots,120\}^6:\sum^6_{k=1} s_k^i = 120\}$, and the payoff of player $i$ is defined by
\begin{equation}\label{eq:payoffs}
    \purePayoff^{i}\left(\pure^{i}, \pure^{-i} \right) = \sum^6_{k=1} \,\mathds{1}_{s_k^{i} > s_k^{-i}}.\notag
\end{equation} 
\end{example} 

\bigskip
\noindent Abstracting from the fact that the first example keeps $N$ and $K$ flexible, the main difference between the two examples lies in the tie-breaking rule applied. In contrast to Example 1, the game in Example 2 fails to be constant-sum, as the sum of payoffs of both players depends on probabilities of 
ties occurring on individual battlefields. This complicates the equilibrium analysis but also makes the game more interesting. Intuitively, setting $\alpha=0$ provides an additional incentive to outguess the opponent. For instance, in an experiment with $\alpha=1$, subjects might perceive the pure strategy $(20,20,20,20,20,20)^{\prime}$ as a focal point that leads to a fair division.\footnote{Here and below, we represent the strategy as a row vector via transposition.} With $\alpha=0$, however, payoffs at this focal point are zero for both players, i.e., there is a strong incentive to engage in at least \textit{some} additional reasoning. As this consideration is mentioned by Arad and Rubinstein (\citeyear{arad_rubinstein}) early in their work, it might have contributed to their decision to depart from the standard tie-breaking rule. 

\bigskip
\noindent \textit{2.3 Assumptions}

\noindent For expositional reasons, we will initially work under a set of simplifying assumptions. The implications of dropping these assumptions will be discussed in the extensions section. Our first assumption concerns the parity of the number of battlefields.

\begin{assumption}\label{ass:1}
    $K$ is even.
\end{assumption}

\noindent Assumption 1 simplifies the analysis but can often be dropped at the cost of additional arguments. Experimental papers tend to work under the assumption. Our second assumption concerns the relationship between $N$ and $K$.

\begin{assumption}\label{ass:2}
    $N$ is divisible by $K$.
\end{assumption}

\noindent Thus, in the main part of the analysis, we will assume that the number of resources $N$ is a multiple of $K$. It follows from Assumption~2 that 
\begin{equation}
m=\frac{N}{K}  \notag  
\end{equation}
is an integer. For an experimental subject, this means that the uniform allocation that assigns $m$ units of the resource to each of the battlefields is a possibility. E.g., Assumption 2 is violated in Avrahami and Kareev (\citeyear{avrahami2014allocation}), but it holds in Avrahami et al.~{(\citeyear{avrahami2014allocation}}). The parameter $m$ will play an important role in the sequel. The same is true for the efficiency parameter $\alpha$ on which we impose the following assumption.

\begin{assumption}\label{ass:3}
    $\alpha\in [0,2]$.
\end{assumption}

\noindent Assumption 3 may be considered natural but nevertheless restricts the efficiency parameter in two ways. The restriction $\alpha\geq0$ says that ties cannot cause inefficiencies beyond the complete loss of the battlefield value. For $\alpha<0$, players would have a very strong incentive to avoid ties, which may lead to asymmetric equilibria. Conversely, for $\alpha>2$, ties would be ``superefficient,'' so that pure-strategy equilibria become natural. Such possibilities will be further discussed in the extensions section.

Note that Assumptions 1 through 3 hold in the Arad-Rubinstein game. Indeed, in Example 2, $K=6$ is even, $m=\frac{120}{6}=20$ is an integer, and $\alpha=0\in [0,2]$.  

\bigskip
\noindent \textit{2.4 Equilibrium concept}

\noindent Given a finite non-empty set $X$, let $\dist\left(X\right)$ denote the set of all probability distributions on~$X$. We are interested in \textit{mixed-strategy Nash equilibria} of the game, i.e., mixed strategy profiles $\sigma = (\sigma^{i}, \sigma^{-i}) \in \dist(S) \times \dist(S)$ such that no player can improve her expected payoff by unilaterally changing her mixed strategy. By a \textit{symmetric Nash equilibrium strategy}, we mean any mixed strategy $\sigma^i$ such that $\sigma=(\sigma^i,\sigma^i)$ is a mixed-strategy Nash equilibrium. 

A crucial property of Colonel Blotto games is that payoffs are functions of the respective marginal distributions at each battlefield. Given a mixed strategy $\sigma^i$ and a battlefield $k\in\{1,\ldots,K\}$, let $\sigma^i_k$ denote the marginal distribution of $\sigma^i$ at battlefield $k$. Following Hart (\citeyear{hart}), we denote the uniform marginal on $\{0,\dots,2m\}$ by $U^m$. A mixed strategy $\sigma^i$ such that $\sigma^{i}_k =U^m$ for every battlefield $k$ will be said to induce \textit{uniform marginals}. 


\newpage
\noindent \textbf {3. Equilibrium in the Colonel Blotto game with flexible tie-breaking}
\medskip

\noindent In this section, we present a simple approach that allows constructing an equilibrium in the Colonel Blotto game with the flexible tie-breaking rule. As mentioned before, this will lead us to an equilibrium in the Arad-Rubinstein game as well. 

\bigskip
\noindent \textit{3.1 A canonical equilibrium}

\noindent The following result characterizes one particular Nash equilibrium in $\mathcal{B}_{\alpha}(N,K)$.

\begin{proposition}\label{prop:player_symm_NE}
Impose Assumptions 1 through 3. Then, a symmetric equilibrium strategy of $\mathcal{B}_{\alpha}(N,K)$ is given by uniform randomization over the set of pure strategies
\begin{equation}
     S_0 = 
     \left\{ {\left( 
\begin{matrix}
0 \\ 
2m \\ 
\vdots \\
0 \\
2m
\end{matrix}%
\right)} , 
{\left( 
\begin{matrix}
1 \\ 
2m - 1 \\ 
\vdots \\
1 \\
2m - 1
\end{matrix}%
\right)} 
, \ldots ,
{\left( 
\begin{matrix}
2m \\ 
0 \\ 
\vdots \\
2m \\
0
\end{matrix}%
\right)}
\right\}.\notag
\end{equation}
In the resulting equilibrium, players' expected payoffs amount to 
    $\pi^* = K\cdot\frac{m+ \frac{\alpha}{2}}{2m+ 1}$.
\end{proposition}
\begin{proof} See the Appendix. 
\end{proof}

\noindent As can be seen, players partition the set of battlefields into pairs. This is possible, of course, because the number of battlefields has been assumed even via Assumption 1.\footnote{We would like to add, however, that Assumption 1 can be dropped, as will be explained in the extensions section.} To each pair of battlefields, a constant number of $2m$ resources is allocated, which is feasible due to Assumption 2. Moreover, the split among the two battlefields in each pair is uniformly distributed, with perfect correlation across pairs. We note that players' strategies induce uniform marginals. 
Indeed, for any battlefield $k$, every number of resources in $\{0,1,2,\ldots, 2m\}$ is assigned to battlefield $k$ with the same probability of $\frac{1}{2m+1}$.

\bigskip
\noindent \textit{3.2 Discussion}

\noindent Hart~(2008) has shown that, under Assumptions 1 and 2, strategies inducing uniform marginals form an equilibrium in the constant-sum game $\mathcal{B}_{1}(N,K)$. The point to note is that this remains the case even with flexible tie-breaking. The reason why the equilibrium property does not break down with more flexible tie-breaking, and this is the main observation that motivated our work on the present paper, is that, given uniform marginals, \textit{the likelihood of getting tied is constant across all bid levels that are used in equilibrium with positive probability}. As a result, the indifference in the standard setup with $\alpha=1$ is not affected by the modification of the payoff functions at ties.

One might wonder if, with the modified tie-breaking in place, players would not have an incentive to bid higher than $2m$ on some of the battlefields. Such an incentive might arise for $\alpha<1$, because the tie-breaking is inefficient in that case. However, overbidding cannot raise a player's payoff. The reason is that any additional unit of the resource required to assign more than $2m$ on some battlefield needs to be taken from some other battlefield, where this lowers the probability of winning by $\frac{1}{2m+1}$. Indeed, every bid level in $\{0,1,\ldots,2m\}$ is used by the opponent on every battlefield with the same probability of $\frac{1}{2m+1}$. On the other hand, the increase in payoff from bidding $2m+1$ instead of $2m$ in a battlefield is $\frac{1}{2m+1}\times(1-\alpha)\leq \frac{1}{2m+1}$.
Therefore, given Assumption 3, or more precisely given that $\alpha\geq 0$, the deviation never yields a strictly higher payoff.\footnote{As the discussion shows, the conclusion of Proposition 1 remains technically true for $\alpha>2$. However, as mentioned before, there is no reason to hide one's strategy for $\alpha\geq 2$, i.e., symmetric pure-strategy equilibria may be more plausible in that case.}

\bigskip
\noindent \textit{3.3 Illustration}

\noindent We illustrate Proposition 1 with an example.

\begin{corollary}[\textbf{Equilibrium in the Arad-Rubinstein game}]\label{cor:unifrom}
The following strategy is a symmetric Nash equilibrium strategy in $\mathcal{B}_0(120,6)$. Both players individually and independently randomize uniformly over the set
\begin{equation}
     S_0= 
     \left\{ {\left( 
\begin{matrix}
0 \\ 
40 \\ 
0 \\ 
40 \\ 
0 \\
40
\end{matrix}%
\right)} , 
{\left( 
\begin{matrix}
1 \\ 
39 \\ 
1 \\ 
39 \\ 
1 \\
39
\end{matrix}%
\right)} 
, \ldots ,
{\left( 
\begin{matrix}
40 \\ 
0 \\ 
40 \\ 
0 \\ 
40 \\
0
\end{matrix}%
\right)}
\right\}\text{.}\notag
\end{equation}
The equilibrium payoff is 
$\pi^*=\frac{120}{41}\approx 2.927$.
\end{corollary}

\begin{proof}
    Immediate from Proposition 1.
\end{proof}

\noindent Even though the equilibrium characterized by Proposition 1 has a canonical structure (e.g., it has uniform marginals, and is symmetric with respect to permutations of the battlefield pairings and within battlefield pairings), it is not the only equilibrium, as will be shown in the next section.

\bigskip\bigskip
\noindent \textbf{4. Understanding the equilibrium set}
\medskip

\noindent The equilibrium set turns out to be very large. To illustrate this point, we study the support of equilibrium strategies in the present section. We first identify pure strategies that are chosen with positive probability in some mixed-strategy Nash equilibrium (Subsection 4.1). Then, we outline the proof (Subsection 4.2). Finally, we turn to strategies that are never ``good'' (Subsection 4.3).

\bigskip \noindent 
\textit{4.1 Pure strategies that arise in some equilibrium}

\noindent The following result provides an indication about the size of the equilibrium set.

\begin{proposition}\label{prop:universality}
Impose Assumptions 1 through 3. Then, every pure strategy $s^i$ such that $s^i_k \leq \frac{2N}{K}$ for every battlefield $k$, is used with positive probability in some equilibrium strategy of $\mathcal{B}_{\alpha}(N,K)$.
\end{proposition}

\begin{proof}
    See Subsection 4.2. 
\end{proof}

\noindent Thus, every pure strategy that does not allocate an excessive number of resources to an individual battlefield is part of some mixed-strategy Nash equilibrium.

In the Arad-Rubinstein game $\mathcal{B}_{0}(120,6)$, every pure strategy that assigns at most 40 soldiers to each of the battlefields is part of some mixed strategy Nash equilibrium. This observation illustrates a drawback of the Nash equilibrium concept for the analysis of Colonel Blotto games with practical relevance. The set of equilibrium predictions is simply very large. We will come back to this issue in the next section.

Proposition 2 relates to observations made by Tukey (\citeyear{tukey1949problem}) saying that ``there are good strategies in which a given player
either (i) sends out no units, (ii) sends out more than half of some kind
of unit, or (iii) sends units to more than half of the available sites.''\footnote{These general observations hold, in particular, for an asymmetric version of the Colonel Blotto game discussed in McDonald and Tukey (\citeyear{mcdonald1949colonel}).} Interpreting ``good'' as appearing in the support of an equilibrium strategy, it is not hard to see that Proposition 2 implies conditions (i) and (iii) under the assumptions of the present paper. An illustration of the possibility of condition (ii) will be given later in the paper (see Example \ref{ex:3}).

\bigskip
\noindent\textit{4.2 Proof of Proposition 2}

\noindent To understand why Proposition~\ref{prop:universality} is true, suppose that both players uniformly randomize over the set of pure strategies
\begin{equation*}
S_{1}=\left\{ \left( 
\begin{matrix}
s_{1} \\ 
2m-s_{1} \\ 
\vdots \\ 
s_{L} \\ 
2m-s_{L}%
\end{matrix}%
\right) :s_{1},\ldots ,s_{L}\in \{0,1,\ldots ,2m\}\right\} .
\end{equation*}%
Clearly, this strategy induces uniform marginals.
But this implies, by the discussion following Proposition 1, that both players are actually using an equilibrium strategy.\footnote{However, compared to Proposition 1, the correlation of the $L$ uniform distributions, one for each pair of battlefields, has been dropped.} 

Let $s$\textit{\ }be any pure strategy such that\textit{\ }$s_{k}\leq 2m$ for all $k\in\{1,\ldots,K\}$.
Then, it suffices to replace the two pure strategies%
\begin{equation*}
\left( 
\begin{matrix}
s_{1} \\ 
2m-s_{1} \\ 
\vdots \\ 
s_{K-1} \\ 
2m-s_{K-1}%
\end{matrix}%
\right) \text{, }\left( 
\begin{matrix}
2m-s_{2} \\ 
s_{2} \\ 
\vdots \\ 
2m-s_{K} \\ 
s_{K}%
\end{matrix}%
\right)
\end{equation*}%
in the support of the mixed strategy $\widehat{\sigma }$ by 
\begin{equation}
\left( 
\begin{matrix}
s_{1} \\ 
s_{2} \\ 
\vdots \\ 
s_{K-1} \\ 
s_{K}%
\end{matrix}%
\right) \text{, }\left( 
\begin{matrix}
2m-s_{2} \\ 
2m-s_{1} \\ 
\vdots \\ 
2m-s_{K} \\ 
2m-s_{K-1}%
\end{matrix}%
\right) \text{,}\notag
\end{equation}
respectively. The marginals do not change, and hence, we have found an equilibrium strategy in which $s$ is played with positive probability. This concludes the argument.

\bigskip
\noindent\textit{4.3 Strategies that are never ``good''}

\noindent In analogy to Proposition 2, one may ask what type of pure strategies are never used in any equilibrium. To address this question, we identify the maximum loss of efficiency that is feasible with modified tie-breaking. Then, we derive conditions on pure strategies that make them render an expected payoff too low to correspond to the maximum efficiency loss. Proceeding along these lines, we show that any strategy that focuses on too few battlefields will never be part of any mixed-strategy Nash equilibrium. 


\begin{proposition}\label{prop:unused_strategies}
Impose Assumptions 1 and 2, and let $\alpha \in [0,1)$. Then, any pure strategy that allocates the resource to strictly less than 
\begin{equation*}
    K^{\ast} = 
    \frac{NK}{2N+K}
\left(1-\frac{\alpha}{2-\alpha}
\right)
\end{equation*} 
battlefields is never part of any equilibrium of $\mathcal{B}_{\alpha}(N,K)$.
\end{proposition}

\begin{proof}
    See the Appendix. 
\end{proof}
\noindent 
How strong is the conclusion of Proposition~\ref{prop:unused_strategies} in specific games? For the Arad-Rubinstein game, 
\begin{equation}
K^{\ast}= 
\frac{120\cdot 6}{2\cdot 120 +6} 
= 2.9268.    \notag
\end{equation}
This means that every pure strategy that assigns a positive number of resources to less than three battlefields is never a part of any equilibrium.

\begin{corollary}
    In the Arad-Rubinstein game, any pure strategy that allocates the resource to just one or two battlefields is never part of any equilibrium.
\end{corollary} 

\begin{proof}
See the text above. 
\end{proof}
\noindent Thus, pure strategies such as $(120,0,0,0,0,0)^{\prime}$ and $(60,60,0,0,0,0)^{\prime}$ are never good strategies in the example. Regrettably, this leaves a gap to the conclusion of Proposition~\ref{prop:universality}. I.e., we do not know if strategies such as $(60,30,30,0,0,0)^{\prime}$ that assign strictly positive resources over at least three battlefields and more than 40 soldiers to at least one of those are ``good.''\footnote{As $\alpha$ becomes larger, the conclusion of Proposition \ref{prop:unused_strategies} weakens. For instance, the conclusion becomes void in the limit case as $\alpha\rightarrow 1$.}

\newpage
\noindent \textbf{5. Refinements}
\medskip



\noindent Given that the Arad-Rubinstein game is a static game, it is natural to check if dominance relationships between strategies might help to narrow down the set of equilibria (e.g., Kohlberg and Mertens,~\citeyear{kohlberg1986strategic}). 

Recall that a pure strategy $s^i \in S$ for player $i\in\{A,B\}$ is \textit{weakly dominated} by another pure strategy $\widehat{s}^i \in S$ if (i) $\pi^i(s^i,s^{-i})\leq \pi^i(\widehat{s}^i,s^{-i})$ for every pure strategy $s^{-i}\in S$, and (ii) there exists at least one pure strategy $s^{-i}\in S$ for the opponent such that $\pi^i(s^i,s^{-i})< \pi^i(\widehat{s}^i,s^{-i})$. Thus, a pure strategy is weakly dominated by another pure strategy if it never yields a greater payoff than the other strategy, but a strictly lower payoff than the other strategy against at least one pure strategy of the opponent. 

It turns out that the elimination of dominated strategies by pure strategies is entirely ineffective for small $\alpha$.

\begin{proposition}\label{prop:WeakDominance}
Suppose that $\alpha < \frac{2}{K}$. Then, no pure strategy in $\mathcal{B}_{\alpha}(N, K)$ is weakly dominated by any other pure strategy.
\end{proposition}

\begin{proof}
See the Appendix.
\end{proof}

\noindent The condition in the proposition is satisfied, in particular, in the Arad-Rubinstein game. There, given that identical choices of pure strategies by the two players lead to ties in all battlefields, $\alpha=0$ implies that the diagonal entries of the payoff matrix are all zero. In contrast, all of the off-diagonal entries of the payoff matrix are positive because at least one battlefield is won by each player if strategies differ. Therefore, no pure strategy is weakly dominated by any other pure strategy if $\alpha=0$. This idea of the proof generalizes in a straightforward way to positive but sufficiently small $\alpha$.\footnote{However, the conclusion of Proposition 4 need not hold for standard tie-breaking. As we are going to show in Section 8, a pure strategy may be weakly dominated by another pure strategy if $\alpha=1$.}

\begin{corollary}
    In the Arad-Rubinstein game, there are no strategies that are weakly dominated by a pure strategy.
\end{corollary}

\begin{proof}
    Immediate from Proposition \ref{prop:WeakDominance}.
\end{proof}


\bigskip\bigskip
\noindent \textbf{6. Adaptive learning}
\medskip

\noindent We have seen above that the traditional game-theoretic analysis of the Arad-Rubinstein game is bound to remain inconclusive. Although a small number of pure strategies could be categorized as ``bad,'' the concept of Nash equilibrium, even after applying standard refinements, is not sufficiently stringent to derive a meaningful reference point for empirical work. In this section, we therefore approach the problem of making a theory-led prediction from a different perspective, viz.\ by determining the implications of long-run adaptive learning. 

The biggest obstacle to a simulation of long-run adaptive learning in a Colonel Blotto game is the size of the strategy space. The number of pure strategies in the Arad-Rubinstein game is \begin{equation}
    \vert S \vert = \binom{125}{5} = 234^{\prime}531^{\prime}275\text{.} \notag
\end{equation}
To be able to obtain results within a reasonable time frame, we decided to exploit the symmetry of the game. Technically, this amounts to considering the set of pure strategies in the corresponding Colonel Lotto game. Thereby, the size of the strategy space reduces to $436^{\prime}140$ pure strategies.\footnote{Details on the computation of the number of strategies are provided in the Appendix.} 

For adaptive learning, we assumed that players follow fictitious play (Brown, \citeyear{brown1949some}; Robinson, \citeyear{robinson1951iterative}).\footnote{For the computation of equilibrium, more efficient algorithms are available (see, e.g., Ahmadinejad et al.,~\citeyear{ahmadinejad2019duels}). However, our aim here is the simulation of a learning process.} Thus, after an initial period of play, each player considers the empirical frequency distribution of prior play as the best predictor for future play. In our simulation, players interacted over 50 million rounds. The learning algorithm has been implemented in Matlab. The computations were conducted on a \textit{NEMA cluster}, which is a high-performance SGI Altix UV2000 system, equipped with 4TBs of main memory and 96 physical CPU cores. The operating system was Unix. \medskip


\begin{figure}[h]
    \centering
    \includegraphics[width=0.85\linewidth]{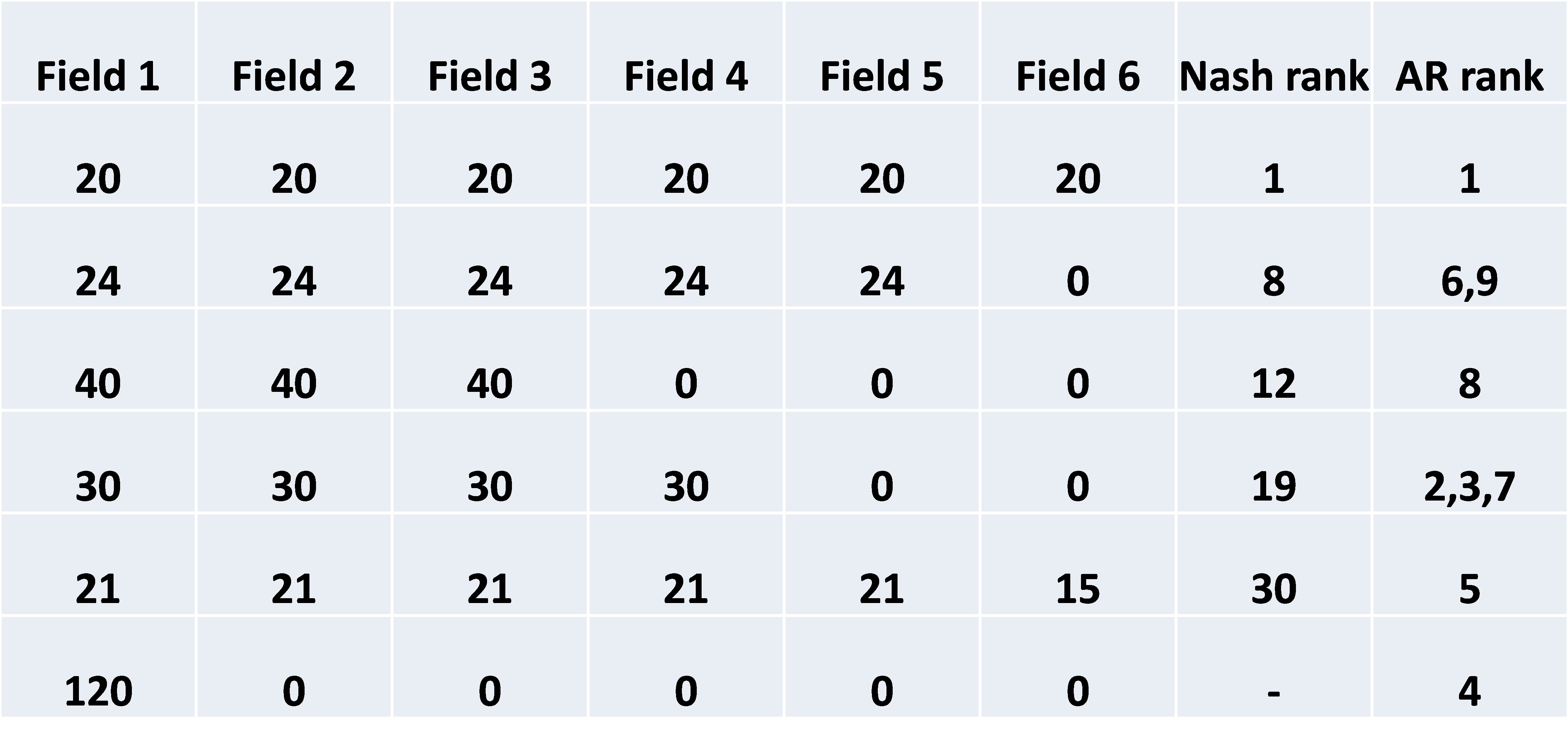}
    \caption{Rank-order analysis}
    \label{fig:enter-label}
\end{figure}


The
support of the learned mixed strategy in the Colonel Lotto game consisted of $\vert S_0\vert =11^{\prime}008$ pure strategies. Table I shows the nine most likely pure strategies in
the learned mixed strategy and compares their respective rank in the learned strategy (column ``Nash rank'') to the rank found by Arad and
Rubinstein (\citeyear{arad_rubinstein}, Table 9). As can be seen, the rankings are related.
Specifically, eight out of the top-9 ordered strategies appearing in the
Arad-Rubinstein data are among the top-30 unordered Nash strategies (the
exception being the rather odd strategy $(120,0,0,0,0,0)^{\prime}$). Moreover, the modal
strategy $(20,20,20,20,20,20)^{\prime}$ is identical between the two rankings. 
As for Arad and Rubinstein's (\citeyear{arad_rubinstein}) \textquotedblleft winning
strategy\textquotedblright\ $(31,31,31,23,2,2)^{\prime}$, it appears at position 47 in the
creation of new best responses during the fictitious-play process and later
ends up at Nash rank 272. 
These observations intuitively square with the
iterative reasoning explanation, and suggest that further inquiry of the relationship between learning and iterative reasoning in multidimensional decision problems might be worth to be pursued.

\bigskip\bigskip
\noindent \textbf{7. Extensions}

\medskip
\noindent This section offers several extensions. We first discuss the case of an odd number of battlefields (Subsection 7.1), then the case where $n$ is not divisible by $K$ (Subsection 7.2), next payoff-inequivalent equilibria (Subsection 7.3), and finally the possibility of pure strategy Nash equilibria (Subsection 7.4).    

\bigskip 
\noindent \textit{7.1 Odd number of battlefields}

\noindent If the number of battlefields is odd so that Assumption 1 is violated, then it is no longer feasible to partition the set of battlefields into pairs. As a result, the equilibrium analysis is complicated substantially. However, using a result of Dziubiński~(\citeyear{dziubinski_2}), the conclusion of Proposition 1 can be shown to hold even if Assumption 1 is dropped. Specifically, even if $K$ is odd, there still exists a mixed strategy that induces uniform marginals on every battlefield. As space limitations make it impossible to replicate the original arguments, we confine ourselves to an illustrative example.\footnote{Further background on this example is provided in the Appendix.}

\begin{example} 
\label{ex:3} 
Consider the Colonel Blotto game $\mathcal{B}_{\alpha}(6,3)$, where $\alpha\in[0,2]$. Then, the mixed strategy $\sigma^i$ given by 
\begin{equation}
\sigma^{i}\left(s^i\right) = 
    \begin{cases}
       \; \dfrac{1}{10} &\text{ if $s^i \in S_a$,} \\
        \\
        \; \dfrac{1}{20} &\text{ if $s^i \in S_b$,} \\
        \\
        \;\;0 &\text{ otherwise,}
    \end{cases}\notag   
\end{equation}
\noindent where
\begin{eqnarray}
     S_a &=&
\left\{ 
\left(
\begin{matrix}
2 \\
2 \\
2
\end{matrix}%
\right),
\left(
\begin{matrix}
3 \\
3 \\
0
\end{matrix}%
\right),
\left(
\begin{matrix}
3 \\
0 \\
3
\end{matrix}%
\right),
\left(
\begin{matrix}
0 \\
3 \\
3
\end{matrix}%
\right),
\left(
\begin{matrix}
4 \\
1 \\
1
\end{matrix}%
\right),
\left(
\begin{matrix}
1 \\
4 \\
1
\end{matrix}%
\right),
\left(
\begin{matrix}
1 \\
1 \\
4
\end{matrix}%
\right)
\right\}
\text{,} \notag
\\
     S_b &=&
\left\{ 
\left(
\begin{matrix}
4 \\
2 \\
0
\end{matrix}%
\right),
\left(
\begin{matrix}
4 \\
0 \\
2
\end{matrix}%
\right),
\left(
\begin{matrix}
2 \\
4 \\
0
\end{matrix}%
\right),
\left(
\begin{matrix}
2 \\
0 \\
4
\end{matrix}%
\right),
\left(
\begin{matrix}
0 \\
4 \\
2
\end{matrix}%
\right),
\left(
\begin{matrix}
0 \\
2 \\
4
\end{matrix}%
\right)
\right\}\text{.}\notag
\end{eqnarray}
 induces uniform marginals and is a symmetric equilibrium strategy.
\end{example}

\noindent Generalizing the construction underlying Example 3, we obtain the following result.

\begin{proposition}\label{prop:Ass1}
The conclusion of Proposition 1 continues to be true if Assumption 1 is dropped.
\end{proposition}

\begin{proof}
    See the Appendix.
\end{proof}


\medskip
\noindent \textit{7.2 $N$ is not divisible by $K$}

\noindent Assumption 2, which has been assumed in the main strand of the analysis, requires that $N$ is divisible by $K$. We can relax this assumption somewhat, as illustrated by the following example 

\begin{example} Impose Assumption 3. Then, in $\mathcal{B}_{\alpha}(6,4)$, the mixed strategy that randomizes uniformly over the set
\begin{equation}
    S_{\sigma^i} = \left\{
    \left(
    \begin{matrix}
        0 \\
        3 \\
        0 \\
        3 
    \end{matrix}
    \right), 
        \left(
    \begin{matrix}
        1 \\
        2 \\
        1 \\
        2 
    \end{matrix}
    \right), 
    \left(
    \begin{matrix}
        2 \\
        1 \\
        2 \\
        1 
    \end{matrix}
    \right), 
    \left(
    \begin{matrix}
        3 \\
        0 \\
        3 \\
        0 
    \end{matrix}
    \right) 
    \right\},\notag
\end{equation}
is a symmetric equilibrium strategy.
\end{example}
\noindent It is not hard to convince oneself that the logic of the example, or equivalently, the proof of Proposition 1, extends to any game in which $K=2\cdot L$ is even and $N$ is divisible by $L$, which is the number of pairs of battlefields.\footnote{We have not attempted to go beyond that case, however.}

\bigskip
\noindent \textit{7.3 Payoff-inequivalent equilibria}

\noindent In the constant-sum model, all equilibria yield the same payoff. This is not the case with flexible tie-breaking, however. To understand why, let $U^m_{\text{O}}$ and $U^m_{\text{E}}$
denote uniform marginals on the odd and even numbers between $0$ and $2m$, respectively. As shown by Hart (\citeyear{hart}, Thm.~7), any pair of strategies that induce marginals in the convex hull of $\{U^m_{\text{O}},U^m_{\text{E}}\}$ is an equilibrium in $\mathcal{B}_{1}(N,K)$. In particular, any strategy profile $\sigma = (\sigma^i, \sigma^{-i})$, where  $\sigma^i_k = U^m_{\text{O}}$ and $\sigma^{-i}_k = U^m_{\text{E}}$ for every battlefield $k$, is an \textit{asymmetric} equilibrium in $\mathcal{B}_{1}(N,K)$. But any such profile remains an equilibrium for $\alpha \in [0,1]$, because inefficient tie-breaking makes deviations less attractive. Moreover, given the absence of ties, equilibrium payoffs are $K/2$ for both players, which is different from the expression in Proposition 1. Similarly, every mixed strategy $\sigma^i$, where  $\sigma^i_k = U^m_{\text{O}}$ or $\sigma^i_k = U^m_{\text{E}}$ for every battlefield $k$ is a \textit{symmetric} equilibrium strategy in $\mathcal{B}_{1}(N,K)$. Since deviations make ties less likely, these profiles remain equilibria for $\alpha \in [1,2]$.

\bigskip
\noindent \textit{7.4 Equilibria in pure strategies}

\noindent Pure-strategy Nash equilibria, both symmetric and asymmetric, are feasible if parameters are outside of the usual range. For instance, if the payoffs from ties are sufficiently high (i.e., if $\alpha$ is close to or exceeds 2), then the Colonel Blotto game transforms into a coordination game with numerous \textit{symmetric} pure strategy equilibria.

\begin{proposition}
\label{prop:PSNE}
If $\alpha \geq \frac{2 \cdot (K - 1)}{K}$, then every pure strategy $s^i\in S$ is a symmetric Nash equilibrium strategy in $\mathcal{B}_{\alpha}(N,K)$. 
\end{proposition}

\begin{proof}
See the Appendix.
\end{proof}

\noindent The pathological outcome indicated by Proposition \ref{prop:PSNE} in cases where $\alpha$ is excessively large provides support for our Assumption 3, a premise upheld throughout the main analysis.

\textit{Asymmetric} pure strategy equilibria may emerge when the number of resources, $N$, is small relative to the number of battlefields (and Assumption 3 is in place). Specifically, if $2N\leq K$, adversaries can easily avoid conflict by dividing the set of battlefields between them.

\bigskip\bigskip
\noindent \textbf{8. Supplementary section: Implications for the standard model}

\noindent Our analysis has also led to new insights for the Colonel Blotto game with standard tie-breaking. These concern the characterization of the set of ``good strategies (Subsection 8.1), a refinement concept (Subsection 8.2), and the effectiveness of weak dominance (Subsection 8.3). 

\bigskip
\noindent\textit{8.1 ``Good'' strategies in the standard model}

\noindent Combining Proposition 2 with a result by Dziubi\'nski (2017), the conclusion can be considerably sharpened in the constant-sum case. 

\begin{corollary}\label{cor:4}
Impose Assumptions 1 and 2. Then, an arbitrary pure strategy $s^i$ is used with positive probability in some equilibrium of $\mathcal{B}_{1}(N,K)$ if and only if $s^i_k \leq \frac{2N}{K}$ for every battlefield $k$.
\end{corollary}

\begin{proof} See the Appendix.     \end{proof}

\bigskip
\noindent\textit{8.2 Refinement}

\noindent Let $\sigma _{\text{even}}$ and $\sigma _{\text{odd}}$ denote mixed
strategies with marginals uniform on the even and odd integers within $%
\{0,\ldots ,2m\}$, respectively. As shown by 
Hart (2008), under Assumptions 1 and 2, (i) such strategies exist
and (ii) any convex combination of such strategies is an equilibrium
strategy in $B_{1}(N,K)$.

\begin{proposition}
\label{prop:robust}
    Impose Assumptions 1 through 3. Then, the only type of strategy that remains an
equilibrium strategy for all values of $\alpha$ in a neighborhood
of $\alpha =1$ has uniform marginals.   
\end{proposition}

\begin{proof}
    See the Appendix. 
\end{proof}

\noindent These observations suggest a simple form of equilibrium selection in finite Colonel Blotto games with standard tie-breaking. The fact that the asymmetric equilibrium $(\sigma_\text{even},\sigma_\text{odd})$ breaks down for $\alpha > 1$ is an indication that this equilibrium is not very robust in the standard model of Hart (in addition to being difficult to coordinate upon, given the asymmetry of the strategy profile). Similarly, the fact that $(\sigma_\text{even},\sigma_\text{even})$ and $(\sigma_\text{odd},\sigma_\text{odd})$ equilibria break down for $\alpha < 1$ is an indication that those equilibria are not very robust either, which might explain why, in the standard model, any equilibrium with uniform marginals, i.e., without parity considerations, is intuitively more plausible than any of the other feasible combinations.


\bigskip
\noindent\textit{8.3 Weak dominance}

\noindent The following example shows that weak dominance may actually eliminate strategies in the Colonel Blotto game with standard tie-breaking.  

\begin{example}
\label{ex:Dom}
In $\mathcal{B}_{1}(120,6)$, the pure strategy $(120, 0, 0, 0, 0, 0)^{\prime}$ is weakly dominated by $(115, 1, 1, 1, 1, 1)^{\prime}$. 
\end{example}

\noindent Details can be found in the Appendix. Example 5 shows that the condition on $\alpha$ imposed in Proposition \ref{prop:WeakDominance} cannot be simply dropped.\footnote{The determination of the set of strategies that are iteratively undominated in the Colonel Blotto game with standard tie-breaking is, however, beyond the scope of the present paper.}
\noindent

\bigskip\bigskip
\noindent \textbf{9. Conclusion}
\medskip

\noindent While the equilibrium analysis of discrete Blotto games is mathematically appealing, it falls short in explaining the observations from applied economic research.
Consequently, beyond merely characterizing the set of Nash equilibria, it becomes imperative to pinpoint ``reasonable'' choices within the expansive strategy set. As suggested by our simulation analysis, exploring the limit points of adaptive learning processes might offer a promising avenue to address this issue. However, it may not be the sole viable approach.

\bigskip\bigskip
\noindent \textbf{A. Appendix}
\renewcommand{\theequation}{A.\arabic{equation}}
\setcounter{equation}{0}
\medskip
\newtheoremstyle{dotless}{}{}{\itshape}{}{\bfseries}{}{ }{}
\theoremstyle{dotless}
\newtheorem{lemma_app}{Lemma A.\hspace{-3px}\!}

\noindent This section contains proofs omitted from the body of the paper. 
The following lemma prepares the proof of Proposition~\ref{prop:player_symm_NE}.




\begin{lemma_app} \label{lem_app:NE_marginals}
Impose Assumptions 2 and 3. Let $\sigma=(\sigma^A,\sigma^B)$ be a mixed strategy profile in $\mathcal{B}_\alpha(N,K)$ such that $\sigma^{i}$ induces uniform marginals for $i\in\{A,B\}$. Then, $\sigma$ is a Nash equilibrium.
\end{lemma_app}

\begin{proof}
Suppose that both players use the assumed mixed strategies that induce uniform marginals. Then, the probability of a tie on any given battlefield is $\frac{1}{2m+1}$. Therefore, each player $i$'s
expected payoff from $\sigma=(\sigma^A,\sigma^B)$ equals
\begin{equation}
    \pi^i(\sigma)=\frac{K}{2} + K \cdot \frac{1}{2m+1} \cdot \frac{\alpha - 1}{2} = \frac{K \cdot \left(2N+ \alpha K\right)}{4N+ 2 K}.\notag
\end{equation}
Suppose that player $i\in\{A,B\}$ deviates to a pure strategy $s^{i}\in S$. The mixed strategy $\sigma^{-i}$ allocates every number of resources in $\{0,1,2,\ldots,2m\}$ on every battlefield $k$ with the same probability of $\frac{1}{2m+1}$. If player $i$ assigns strictly more than $2m$ resources to a battlefield, her payoff from that battlefield is $1$. 
If player $i$ assigns weakly less than $2m$ to battlefield $k$, then she overwhelms the opponent on that battlefield with probability $\frac{s_k^{i}}{2m+1}$ and achieves a tie with probability $\frac{1}{2m+1}$.
Hence, player $i$'s expected payoff from battlefield $k$ is 
\begin{equation}
\pi^{i}_k\left(s^{i},\sigma^{-i}\right)=\min\{1, \frac{s_k^{i}}{2m+1} + \frac{\alpha}{2}\cdot\frac{1}{2m+1}\}.\notag   
\end{equation}
It follows that
\begin{align}
    \pi^{i}\left(s^{i},\sigma^{-i}\right) &= \sum_{k=1}^K \pi^{i}_k\left(s^{i},\sigma^{-i}\right) \notag \\
    &\leq
    \sum_{k=1}^K \left(\frac{s_k^{i}}{2m+1} + \frac{\alpha}{2}\cdot\frac{1}{2m+1}\right) \label{inq} \\
    &= \frac{N}{2m+1} + \frac{K \cdot \alpha}{2 \cdot (2m+1)}\notag \\
    &= \frac{K \cdot (2N+ \alpha K)}{4N+2K}.\notag
\end{align}
As both $i$ and $s^{i}$ were arbitrary, no player can raise her payoff by deviating from $\sigma = (\sigma^{A}, \sigma^{B})$. Hence, $\sigma$ is a Nash equilibrium, which completes the proof. 
\end{proof}

\begin{proof}[\textbf Proof of Proposition 1]
The mixed strategy $\sigma^i$ assigns every number of resources $s_k^{i}\in\{0,1,\ldots,2m\}$ to every battlefield $k$ with the same probability $\frac{1}{2m+1}$. Hence, $\sigma^i$  induces uniform marginals. It therefore follows from Lemma~A.\ref{lem_app:NE_marginals} that $\sigma^i$ is a symmetric equilibrium strategy.
\end{proof}

\begin{proof}[Proof of Proposition~\ref{prop:unused_strategies}]
Consider an arbitrary pure strategy $s^i$ that allocates a positive number of units of the resource to $K_+\in\{1,\ldots, K\}$ battlefields and zero units to the remaining battlefields. Then, player $i$'s expected payoff of $s^i$ against an arbitrary pure strategy $s^{-i}$ satisfies the relationship 
\begin{align}
    \pi^i\left( s^i, s^{-i}\right) &\leq K_+ \cdot 1 + (K-K_+) \cdot \frac{\alpha}{2} \notag\\ 
   & = \frac{\alpha}{2}\cdot K+\frac{2-\alpha}{2}\cdot K_+\text{.} \label{v}
\end{align}
Consider now a deviation by player $i$ to some mixed strategy $\sigma^i$ that induces uniform marginals. Proposition 1 provides an example of such a mixed strategy. The following analysis is analog to the proof of Lemma A.\ref{lem_app:NE_marginals}. If player $-i$ assigns strictly more than $2m$ units of the resource to battlefield $k$, then $i$'s payoff from that battlefield is zero. 
If, however, player $-i$ assigns weakly less than $2m$ to battlefield $k$, then player $i$ overwhelms her opponent on that battlefield with probability $\frac{2m-s_k^{-i}}{2m+1}$ and achieves a tie with probability $\frac{1}{2m+1}$.
Hence, player $i$'s expected payoff from battlefield $k$ is 
\begin{equation}
\pi^{i}_k\left(\sigma^{i},s^{-i}\right)=\max\{0, \frac{2m-s_k^{-i}}{2m+1} + \frac{\alpha}{2}\cdot\frac{1}{2m+1}\}.\notag   
\end{equation}
It follows that
\begin{align}
    \pi^{i}\left(\sigma^{i},s^{-i}\right) &= \sum_{k=1}^K \pi^{i}_k\left(\sigma^{i},s^{-i}\right)
    \notag \\
    &\geq
    \sum_{k=1}^K \left(\frac{2m-s_k^{-i}}{2m+1} + \frac{\alpha}{2}\cdot\frac{1}{2m+1}\right) \notag \\
    &= \frac{NK}{2N+K}+
    \frac{\alpha}{2}\cdot
    \frac{K^2}{2N+K}.\label{u}
\end{align}
A straightforward calculation shows that the right-hand side of equation (\ref{u}) strictly exceeds the right-hand side of equation (\ref{v}) if and only if  
\begin{equation}
K_+<K^\ast \equiv
\frac{2}{2-\alpha}
\cdot \frac{NK(1-\alpha)}{2N+K}
\text{.}\notag
\end{equation}
In particular, in that case, $\pi^{i}\left(\sigma^{i},s^{-i}\right)>\pi^{i}_k\left(\sigma^{i},s^{-i}\right)$, for any $s^{-i}\in S$. Thus, $s^i$ is never a best response if $K_+<K^\ast$. This proves the proposition.\end{proof}

\begin{proof}[\textbf Proof of Corollary \ref{cor:4}]
\textit{(if)} Immediate from Proposition 2. \textit{(only if)} This follows from Dziubi\'nski  (2017, Cor.~2). For the reader's convenience, we offer a direct proof. Consider a mixed-strategy equilibrium $\sigma^\ast = (\sigma^{A,\ast},\sigma^{B,\ast})$ in $\mathcal{B}_{1}(N,K)$. To provoke a contradiction, suppose that there is some player $i\in\{A,B\}$ and some pure strategy $s^i$ in the support of $\sigma^{i,\ast}$ such that $s^i_k>2m$ for some battlefield $k$. As equilibrium strategies in two-person constant-sum game are interchangeable (Osborne and Rubinstein, \citeyear{osborne1994course}, p.~23), $\sigma^{i,\ast}$ is a best response also to the mixed strategy $\sigma^{-i}=\sigma^{i}$ identified in Proposition 1. Moreover, since the bid vector $s^i$ is chosen with positive probability in the mixed strategy $\sigma^{i,\ast}$, the pure strategy $s^i$ is likewise a best response to $\sigma^{-i}$. But from $s^i_k>2m$ for some battlefield $k$, inequality (\ref{inq}) in the proof of Lemma A.\ref{lem_app:NE_marginals} is strict, so that that 
\begin{equation}
\pi^{i}\left(s^{i},\sigma^{-i}\right) < \frac{K \cdot (2N+ \alpha K)}{4N+2K} = \pi^{i}\left(\sigma^{i},\sigma^{-i}\right)\notag.   
\end{equation}
Thus, $s^i$ is not a best response to $\sigma^{-i}$ after all. The contradiction proves the assertion. \end{proof}


\begin{proof}[Proof of Proposition~\ref{prop:WeakDominance}] By contradiction. Suppose that $s$ is a pure strategy that is weakly dominated by another pure strategy $\widehat{s}$. In $\mathcal{B}_{\alpha}(N, K)$, the diagonal entries of the payoff matrix correspond to an outcome with $K$ ties and are, therefore, equal to $K\cdot\frac{\alpha}{2}$. In particular, this is the payoff of $\widehat{s}$ against itself. However, as $\widehat{s}$ is necessarily different from $s$, the strategy $s$ bids strictly higher than $\widehat{s}$ on at least one battlefield. Therefore, the payoff of $s$ against $\widehat{s}$ is at least one. Under the assumption made, this is strictly higher than $K\cdot\frac{\alpha}{2}$. The contradiction shows that no pure strategy can be weakly dominated by any other pure strategy. \end{proof}

\noindent The following lemma is used in the proof of Proposition~\ref{prop:Ass1}.

\begin{lemma_app}[\textbf{Dziubi\'nski, 2017}]\label{lem:dziubinski}
Suppose that Assumption 2 holds and that $K$ is odd. Then, there exists a mixed strategy $\sigma^{i}$ such that $\sigma^{i}_k = U^{m}$ for every battlefield $k\in\{1,\ldots, K\}$. 
\end{lemma_app}

\begin{proof}
    See \citeauthor{dziubinski_2}~(\citeyear{dziubinski_2}, Proposition 2).
\end{proof}

\begin{proof}[Proof of Proposition~\ref{prop:Ass1}]
By Lemma~A.\ref{lem:dziubinski}, we find a uniform strategy for each player even if Assumption 1 fails to hold. As Assumption~1 is not imposed in Lemma~A\ref{lem_app:NE_marginals}, the pair of strategies constitutes a Nash equilibrium in $\mathcal{B}_{\alpha}(N,K)$. This proves the claim.
\end{proof}


\begin{proof}[Details on Example 3.] We will construct a \textit{battlefield-symmetric} uniform equilibrium strategy. Let $p_{411}$ etc.~denote the respective probability that a player chooses the pure strategies $(4,1,1)^{\prime}$ etc. Then, accounting for symmetries, $p_{411}=p_{141}=p_{114}$, etc. Moreover, any solution to the system
\begin{equation}
\left( 
\begin{array}{ccccc}
1 & 0 & 0 & 0 & 2 \\ 
0 & 0 & 2 & 2 & 0 \\ 
0 & 1 & 2 & 0 & 2 \\ 
2 & 0 & 2 & 0 & 0 \\ 
0 & 0 & 0 & 1 & 2%
\end{array}%
\right) \left( 
\begin{array}{c}
p_{411} \\ 
p_{222}\\ 
p_{123} \\ 
p_{330} \\ 
p_{420}%
\end{array}%
\right) =\left( 
\begin{array}{c}
\frac{1}{5} \\ 
\frac{1}{5} \\ 
\frac{1}{5} \\ 
\frac{1}{5} \\ 
\frac{1}{5}%
\end{array}%
\right) \notag
\end{equation}
induces uniform marginals. Restricting attention to solutions that yield nonnegative probabilities, the general solution is given by 
\begin{equation}
\left(\begin{array}{c}
p_{411} \\ 
p_{222}\\ 
p_{123} \\ 
p_{330} \\ 
p_{420}%
\end{array}
\right)= (1-\lambda)
\left(\begin{array}{c}
\frac{1}{10} \\ 
\frac{1}{10} \\ 
0\\ 
\frac{1}{10} \\ 
\frac{1}{20}%
\end{array}
\right)+
\lambda
\left(
\begin{array}{c}
\frac{1}{15} \\ 
0\\ 
\frac{1}{30} \\  
\frac{1}{15} \\ 
\frac{1}{15}%
\end{array}
\right)\text{,} \notag
\end{equation}
where $\lambda\in[0,1]$ is arbitrary. The solution shown in the body of the paper corresponds to $\lambda=0$. 
\end{proof}

\begin{proof}[Proof of Proposition \ref{prop:PSNE}]
Consider a symmetric profile $s = (s^i, s^{-i})$ in pure strategies, where $s^i = s^{-i}$. The payoff to both players from $s$ is equal to $K \cdot \frac{\alpha}{2}$. As player $i$ is unable to win on all the battlefields, the payoff for a deviating player $i$ is bounded from above by $K-1$. Hence, $K - 1 \leq K \cdot \frac{\alpha}{2}$, which is equivalent to the inequality in the statement of the proposition. 
\end{proof}

\begin{proof}[Counting the number of strategies] We start by noting that the number of pure strategies does not depend on the tie-breaking rule. In general, as each pure strategy $s^i\in S$ in the Colonel Blotto game $\mathcal{B}_{\alpha}(N,K)$ corresponds uniquely to a finite sequence 
\begin{equation}
\textstyle
\underbrace{1,\ldots,1}_{s_1^i},\ast, \underbrace{1,\ldots,1}_{s_2^i},\ast, \ldots, \ast, \underbrace{1,\ldots,1}_{s_N^i}  \text{,}   \notag
\end{equation}
where each ``1'' stands for a unit of the resource and the ``$\ast$'' is a separator between neighboring battlefields. Given that there are $N$ units of the resource to allocate and $K-1$ separators, the number of pure strategies in $\mathcal{B}_{\alpha}(N, K)$ is 
\begin{equation}
\vert S \vert =\binom{N+K-1}{K-1}.\notag
\end{equation}
The number of pure strategies in the corresponding Colonel Lotto game, i.e., taking account of symmetries between battlefields, is given by $p(N+K,K)$, where $p(N,K)$ denotes the number of partitions of $N$ into exactly $K$ parts. Although a simple formula is unavailable, 
the recursive relationship 
\begin{equation}
    p(N,K)= p(N-1,K-1) + p(N-K,K)\notag
\end{equation} 
with initial conditions $p(N,1)=1$ and $p(N,K)=0$ if $K>N$ allows computing this number in specific examples (Gupta, \citeyear{gupta1970partitions}).
\end{proof}

\begin{proof}[Proof of Proposition \ref{prop:robust}] The probability of a tie at any given battlefield is zero in $(\sigma _{\text{even}%
},\sigma _{\text{odd}})$, and $(\sigma _{\text{odd}},\sigma _{\text{even}})$%
; similarly, the probability of a tie is $\frac{m+1}{2m+1}$ in $(\sigma _{%
\text{even}},\sigma _{\text{even}})$, and $\frac{m}{2m+1}$ in $(\sigma _{%
\text{odd}},\sigma _{\text{odd}})$. If player $-i$'s marginal is not uniform and $\alpha >1$ ($\alpha <1$), then
player $i$ has the incentive to deviate to the parity that is used more
(less) often by $-i$. \end{proof}

\begin{proof}[Details on Example 5.] Let $s^i=(120,0,0,0,0,0)^{\prime}$. We start by claiming that $s^i$ never results in a payoff higher than 3 against any pure strategy $s^{-i}$. Indeed, if $s^{-i}_1=120$, then player $i$'s payoff is 3. If, however, $s^{-i}_1<120$, then player $i$ wins battlefield 1, loses at least one of the other five battlefields, and achieves at most a tie on the remaining battlefields. Again, therefore, the payoff for player $i$ from choosing $s^i$ cannot exceed 3, which proves the claim. We continue and claim that $s^i$ is weakly dominated by $\widehat{s}^i=(115,1,1,1,1,1)^{\prime}$. For this, we check the conditions in the definition of weak dominance. (i) To provoke a contradiction, suppose that $s^i$ yields a payoff strictly higher than $\widehat{s}^i$ against a given pure strategy $s^{-i}$. Then, $s^i$ necessarily yields a payoff strictly higher than $\widehat{s}^i$ from the first battlefield, meaning that $s^{-i}_1\geq  115$. Hence, $s^{-i}$ assigns at most 5 units of the resource to the remaining 5 battlefields. But each of those units can reduce the payoff from strategy $\widehat{s}^i$ on battlefields 2 through 6 by at most 0.5 utils. Therefore, in the situation considered, $\widehat{s}^i$ yields at least 2.5 units from battlefields 2 through 6. The exact value of 2.5 is only achieved when $s^{-i}_1 = 115$, hence, the resulting payoff is at least 3, as the first battlefield is also tied. When $s^{-i}_1 > 115$, the payoff from battlefields 2 through 6 is at least 3. As follows, player $i$'s payoff from using $\widehat{s}^i$ is at least 3 for every opponent's strategy where $s^{-i} \geq 115$, in conflict with the presumption made above. (ii) Suppose that $s^{-i} = (119,1,0,0,0,0)^{\prime}$. Then, $\pi^i(s^i,s^{-i})=3$, while $\pi^i(\widehat{s}^i,s^{-i})=4.5$. This proves the claim. \end{proof}

\newpage
\renewcommand\refname{\normalsize References}
\bibliographystyle{apalike}
\setlength{\bibhang}{0pt}
\bibliography{paper}

\end{document}